\documentclass[10pt]{article}
\usepackage[utf8]{inputenc}
\usepackage{comment}
\usepackage{graphicx} 
\usepackage{booktabs}
\usepackage[left=1in, bottom=1in, right=1in, top=1in]{geometry}
\usepackage{amsmath, bm}
\usepackage{longtable}
\usepackage{rotating} 
\usepackage{chngcntr}
\usepackage{multirow}
\usepackage{algorithm}
\usepackage{mathrsfs}
\usepackage{xr}
\usepackage{algpseudocode}
\usepackage{amsfonts}
\usepackage[dvipsnames]{xcolor}
\usepackage{enumitem}
\usepackage{wrapfig}
\usepackage{algpseudocode}
\usepackage{pgfplots}
\usepackage{amsthm}
\usepackage{caption}
\usepackage{csquotes}
\usepackage{hyperref}

\newtheorem{theorem}{\noindent Theorem}

\newtheorem{corollary}{\noindent Corollary}

\newtheorem{assumption}{\noindent Assumption}

{}
{}
{}
{}
{\newtheorem{result}{\noindent Result}}
{\newtheorem{remark}{\noindent Remark}}
{}
{}
{}
{ }
{}
{}
{}
{}
{}
{}
{}
{}

\newcommand{\secnd}[1]{\left\{#1 \right\}}

\pgfplotsset{compat = 1.18}
\hypersetup{
	colorlinks=true,
	linkcolor=blue,
	filecolor=magenta,      
	urlcolor=blue,
	citecolor=blue
}

\usepackage{subcaption}
\usepackage{tikz}

\DeclareRobustCommand{\truejump}{%
\tikz[baseline=-0.6ex]
\draw[blue, line width=0.8pt] (0,0) circle (0.7ex);
}

\DeclareRobustCommand{\detectedjump}{%
\tikz[baseline=-0.6ex]
\draw[red, line width=0.9pt]
(-0.7ex,-0.7ex) -- (0.7ex,0.7ex)
(-0.7ex,0.7ex) -- (0.7ex,-0.7ex);
}

\usepackage[citestyle=authoryear-comp,
    maxcitenames=2,
    giveninits=false,
    uniquename=false,
    uniquelist=false,
	ibidtracker=context,
    bibstyle=authoryear,
    dashed=false, 
    date=year,
	sorting=nyt,
    minbibnames=3,
	hyperref=true,
	backref=true,
	citecounter=true,
	citetracker=true,
    natbib=true, 
	backend=biber, 
]{biblatex}

\DeclareFieldFormat{citehyperref}{%
  \DeclareFieldAlias{bibhyperref}{noformat}
  \bibhyperref{#1}}

\DeclareFieldFormat{textcitehyperref}{%
  \DeclareFieldAlias{bibhyperref}{noformat}
  \bibhyperref{%
    #1%
    \ifbool{cbx:parens}
      {\bibcloseparen\global\boolfalse{cbx:parens}}
      {}}}

\savebibmacro{cite}
\savebibmacro{textcite}

\renewbibmacro*{cite}{%
  \printtext[citehyperref]{%
    \restorebibmacro{cite}%
    \usebibmacro{cite}}}

\renewbibmacro*{textcite}{%
  \ifboolexpr{
    ( not test {\iffieldundef{prenote}} and
      test {\ifnumequal{\value{citecount}}{1}} )
    or
    ( not test {\iffieldundef{postnote}} and
      test {\ifnumequal{\value{citecount}}{\value{citetotal}}} )
  }
    {\DeclareFieldAlias{textcitehyperref}{noformat}}
    {}%
  \printtext[textcitehyperref]{%
    \restorebibmacro{textcite}%
    \usebibmacro{textcite}}}

\DeclareAutoCiteCommand{inline}{\textcite}{\parencite}

\title{Asymptotic Separability of Diffusion and Jump Components \\ in High-Frequency CIR and CKLS Models}
\author{Sourojyoti Barick \\ Interdisciplinary Statistical Research Unit, Indian Statistical Institute, Kolkata}       
\begin{document}

\maketitle

\begin{abstract}
This paper develops a robust parametric framework for jump detection in discretely observed CKLS-type jump–diffusion processes with high-frequency asymptotics, based on the minimum density power divergence estimator (MDPDE). The methodology exploits the intrinsic asymptotic scale separation between diffusion increments, which decay at rate $\sqrt{\Delta_n}$, and jump increments, which remain of non-vanishing stochastic magnitude. Using robust MDPDE-based estimators of the drift and diffusion coefficients, we construct standardized residuals whose extremal behavior provides a principled basis for statistical discrimination between continuous and discontinuous components. We establish that, over diffusion intervals, the maximum of the normalized residuals converges to the Gumbel extreme-value distribution, yielding an explicit and asymptotically valid detection threshold. Building on this result, we prove classification consistency of the proposed robust detection procedure: the probability of correctly identifying all jump and diffusion increments converges to one under proper asymptotics. The MDPDE-based normalization attenuates the influence of atypical increments and stabilizes the detection boundary in the presence of discontinuities. Simulation results confirm that robustness improves finite-sample stability and reduces spurious detections without compromising asymptotic validity. The proposed methodology provides a theoretically rigorous and practically resilient robust approach to jump identification in high-frequency stochastic systems.
\end{abstract}
\section{Introduction}

Stochastic differential equations (SDEs) constitute a fundamental mathematical
framework for modeling dynamically evolving systems subject to random
perturbations. In financial econometrics, such models play a central role in the
analysis of interest rate dynamics, asset prices, and volatility processes.
Among the most prominent specifications, the Cox--Ingersoll--Ross (CIR) process,
introduced by \citet{cox1985}, has received considerable attention due to its
ability to capture mean-reverting behavior while preserving strict positivity
of the state variable. The CIR model arises as a special case of the more
general Chan--Karolyi--Longstaff--Sanders (CKLS) model proposed by
\citet{ckls1992}, which provides a flexible parametric family capable of
accommodating empirically observed features such as
state-dependent volatility, and heterogeneous diffusion elasticities.

The empirical relevance of mean–reverting diffusion processes in interest rate modeling is well established. Short-term interest rates exhibit a persistent tendency to revert toward a stochastic equilibrium level, reflecting macroeconomic fundamentals, monetary policy interventions, and market expectations. Such dynamics are naturally represented by continuous-time diffusion models, which provide a parsimonious and analytically tractable description of short-rate evolution. These models form the structural foundation of modern arbitrage-free term structure theory and have been systematically incorporated within general no-arbitrage frameworks; see, for example, \citet{Heath1992} and \citet{DuffieSingleton1993}.  Since derivative valuation, hedging, and risk management depend explicitly on the drift and diffusion coefficients governing the underlying rate dynamics, reliable statistical inference for discretely observed diffusions is essential. This has motivated extensive research on statistically efficient estimation procedures, most notably maximum likelihood estimation (MLE) and conditional least squares (CLS), which provide consistent and asymptotically normal estimators under high-frequency sampling schemes.

Early contributions to the statistical inference of the CIR process include the
work of \citet{Overbeck1997}, who proposed CLS-type estimators and established
their strong consistency under both discrete and continuous observation schemes.
Subsequent investigations, including \citet{benalya2012} and \citet{benalya2013},
demonstrated that likelihood-based methods generally exhibit superior finite-sample
and asymptotic efficiency relative to CLS estimators, particularly under high-frequency
sampling regimes. These studies established consistency, asymptotic normality,
and improved efficiency properties of maximum likelihood estimators, especially
in regimes where the process approaches boundary regions. Extensions to the
more general CKLS framework were provided by \citet{LiMa2015}, who derived
asymptotic properties of CLS estimators under relaxed regularity conditions and
more general diffusion structures. Alternative inferential paradigms, including
Bayesian estimation procedures incorporating prior structural information, were
developed by \citet{DeRossi2010}, offering enhanced inferential stability in
small-sample environments.

Despite their widespread applicability, classical diffusion models rely
fundamentally on the assumption that the underlying stochastic process evolves
continuously over time. However, empirical evidence from high-frequency
financial data unequivocally demonstrates the presence of discontinuities, or
\emph{jumps}, arising from macroeconomic announcements, liquidity shocks,
institutional trading, and other structural market events. These observations
have motivated the development of jump--diffusion models, which augment the
continuous diffusion component with a discontinuous pure-jump process; see,
for example, the seminal contributions of \citet{Merton1976} and
\citet{Duffie2000_AJD}.

From a statistical standpoint, the presence of jumps introduces profound
challenges for inference. Under high frequency asymptotics, diffusion-driven increments
scale at rate $\mathcal{O}(\sqrt{\Delta t})$, whereas jump-induced increments
remain of order $\mathcal{O}(1)$. This fundamental disparity in asymptotic
scaling provides the theoretical basis for statistical discrimination between
continuous and discontinuous components, and has motivated a substantial body
of literature devoted to jump detection and decomposition of semimartingale
dynamics.

Pioneering work by
\cite{Barndorff2004,Barndorff2006} introduced
nonparametric jump detection procedures based on realized variation measures,
exploiting the distinct asymptotic behavior of quadratic and bipower variation
in the presence of jumps. These methods were further developed by
\citet{Jacod2008} and \citet{AitSahaliaJacod2009}, who established rigorous
limit theorems and consistent statistical tests for detecting discontinuities
under general semimartingale models. A particularly influential contribution is
the locally normalized jump detection statistic proposed by
\citet{LeeMykland2008}, which converges weakly to a standard Gaussian
distribution in the absence of jumps while exhibiting divergent behavior in the
presence of discontinuities. This framework provides a statistically tractable
and asymptotically justified mechanism for identifying jump occurrences in
high-frequency observations.

Notwithstanding these advances, most existing jump detection methodologies are
formulated within a fully nonparametric framework and do not explicitly exploit
the structural information inherent in parametric diffusion models. In many
practical applications, parametric models such as the CIR and CKLS processes
remain indispensable due to their interpretability, tractability, and relevance
for structural and forecasting purposes. A common empirical strategy consists
of applying nonparametric jump detection procedures as a preliminary filtering
step, followed by parametric estimation based on the filtered sample. However,
such two-stage procedures introduce additional stochastic variability and are
inherently vulnerable to misclassification errors, which may propagate into the
parameter estimation stage and degrade statistical efficiency.

A central inferential challenge arises from the intrinsic sensitivity of
likelihood-based estimators to contamination by jump-induced increments.
Classical likelihood and quasi-likelihood methods assign quadratic penalties to
large deviations under Gaussian approximations, rendering them highly sensitive
to atypical observations. Consequently, even a relatively small number of jumps
may exert disproportionate influence on the resulting parameter estimates,
leading to bias, inefficiency, and instability in inference. This phenomenon has
been rigorously analyzed in the context of discretely observed diffusions; see,
for example, \citet{Yoshida1988}, \citet{Kessler1997}, and
\citet{UchidaYoshida2012}.

These limitations underscore the necessity of robust statistical methodologies
capable of mitigating the influence of contamination while retaining asymptotic
efficiency under the correctly specified model. A principled and theoretically
well-founded approach is provided by divergence-based estimation methods, which
replace the classical likelihood function with a discrepancy functional between
the empirical and model-implied distributions. In particular, the minimum
density power divergence estimator introduced by \citet{Basu1998} provides a
continuous interpolation between maximum likelihood estimation and highly robust
procedures through a tuning parameter controlling the trade-off between
efficiency and robustness.

Density power divergence estimators possess several desirable theoretical
properties, including consistency, asymptotic normality, and bounded influence
functions, rendering them particularly suitable for inference in the presence of
contamination or model misspecification. Their statistical properties have been
extensively investigated in both independent and dependent data settings; see,
for example, \citet{Ghosh2013} and \citet{LeeSong2013}. In the context of
discretely observed diffusion processes, divergence-based methods provide a
natural and theoretically justified mechanism for attenuating the influence of
jump-induced increments, which manifest as localized deviations from the
continuous diffusion dynamics.

Motivated by these considerations, the present paper develops a unified
parametric framework for robust jump identification and parameter estimation in
CKLS-type diffusion models. The proposed methodology integrates robust
divergence-based estimation with a statistically principled jump identification
mechanism derived from the asymptotic behavior of locally normalized residuals.
This approach provides a coherent inferential framework that simultaneously
achieves robust parameter estimation and consistent jump detection within a
fully parametric setting.

The remainder of the paper is organized as follows. Section~\ref{PrelimTheo}
reviews the parametric CKLS model and summarizes existing results on consistency
and asymptotic normality of classical estimators. Section~\ref{rob_est_param}
introduces the proposed robust estimation framework and establishes its
asymptotic properties. Section~\ref{test_for_jump} develops a statistically
rigorous jump detection procedure and derives the associated asymptotic
critical thresholds. Section~\ref{Simu} presents simulation studies illustrating
the finite-sample performance of the proposed methodology. Concluding remarks
are provided in Section~\ref{concl}. All technical proofs and auxiliary results
are deferred to Appendix~\ref{sec:app}, with additional simulation results provided
in the Supplementary Material.
\section{Theoretical Background and Illustration}\label{PrelimTheo}

This section summarizes the probabilistic framework and classical asymptotic
results that form the benchmark for the robust methodology developed later.
All statements are standard and included only to fix notation and clarify the
reference model. Detailed derivations can be found in the cited literature.

We consider the CKLS diffusion process introduced by
\citet{ckls1992},
\begin{equation}
dX_t = (\beta_1 - \beta_2 X_t)\,dt + \sigma X_t^\gamma\, dW_t,
\label{eqn:ckls}
\end{equation}
where $\gamma \in [1/2,1]$ and $(W_t)$ is a standard Brownian motion.Throughout the paper we impose the following standing conditions.

\begin{assumption}\label{ass:A1}
The parameters satisfy $\beta_1>0$, $\beta_2>0$, and $\sigma>0$, and the
initial condition satisfies $X_0>0$.
\end{assumption}

The ergodic and boundary properties of \eqref{eqn:ckls} are well understood; see
\citet{ckls1992}, \citet{Karlin1981ASC}, and \citet{andersenPit}. Under
Assumption~\ref{ass:A1}, the process admits a unique stationary distribution.
The boundary at zero is unattainable for $\gamma \in (1/2,1]$, while in the
boundary case $\gamma=1/2$ unattainability holds provided $2\beta_1>\sigma^2$.
The stationary density takes the form
\begin{equation}
f(r)=C(\gamma)\,r^{-2\gamma}\exp\{Q(r;\gamma)\},
\label{eqn:stationary}
\end{equation}
where $Q(\cdot;\gamma)$ is determined by the scale and speed measures.

To motivate the estimation framework, consider first the CIR specification
($\gamma=1/2$),
\begin{equation}
dX_t = (\beta_1 - \beta_2 X_t)\,dt + \sigma \sqrt{X_t}\, dW_t,
\label{eqn:cir}
\end{equation}
introduced by \citet{cox1985}. Let $\{X_{t_i}\}_{i=0}^n$ denote observations on
an equidistant grid with mesh $\Delta_n$. The Euler--Maruyama approximation
(\citealp{KloedenPlaten1992}) yields
\[
X_{t+\Delta_n}
=
X_t + (\beta_1-\beta_2 X_t)\Delta_n
+ \sigma \sqrt{X_t}\Delta W_t.
\]

Define
\[
y_t=\frac{X_{t+\Delta_n}-X_t}{\sqrt{X_t\Delta_n}}, \qquad
z_{1t}=\frac{\sqrt{\Delta_n}}{\sqrt{X_t}}, \qquad
z_{2t}=-\sqrt{X_t\Delta_n}.
\]
Then the discretized model admits the linear representation
\begin{equation}
y_t
=
\beta_1 z_{1t}
+
\beta_2 z_{2t}
+
\varepsilon_t,
\qquad
\varepsilon_t=\frac{\sigma}{\sqrt{\Delta_n}}\Delta W_t,
\label{eqn:lin_reg}
\end{equation}
which forms the basis for classical inference in discretely observed
diffusions; see \citet{Kessler1997} and \citet{Dehtiar2022}.
Let $\hat{\boldsymbol{\beta}}_n$ denote the OLS estimator associated with
\eqref{eqn:lin_reg}, and let $\widehat{\sigma}_n^2$ denote the corresponding
residual variance estimator.

We consider the high-frequency asymptotic regime
\begin{equation}
\Delta_n \to 0,
\qquad
n\Delta_n \to \infty,
\label{eqn:hf}
\end{equation}
under which the observation grid becomes increasingly dense as the sample
size grows. This framework corresponds to an \textit{infill} asymptotic scheme, where
the time discretization is progressively refined while the observation horizon
diverges.

Under~\eqref{eqn:hf}, standard martingale arguments imply consistency and
asymptotic normality of the estimators.
\begin{result}[Consistency]\label{lem:consistency}
\(\hat{\boldsymbol{\beta}}_n \xrightarrow{\mathbb{P}} \boldsymbol{\beta},
\qquad
\widehat{\sigma}_n^2 \xrightarrow{\mathbb{P}} \sigma^2.
\)
\end{result}

\begin{result}[Asymptotic normality]\label{lem:asymp}
\(\frac{\sqrt{n\Delta_n}}{\widehat{\sigma}_n}
\left(
\hat{\boldsymbol{\beta}}_n-\boldsymbol{\beta}
\right)
\xrightarrow{\mathcal{D}}
\mathcal{N}\!\left(0,\Sigma^{-1}\right),
\)
where $\Sigma$ depends on moments of the stationary distribution; see
\citet{Kessler1997}.
\end{result}

For the CIR case,
\[
\Sigma =
\begin{pmatrix}
\frac{\beta_2}{\beta_1 - \sigma^2/2} & -1 \\
-1 & \frac{\beta_1}{\beta_2}
\end{pmatrix},
\]
while for the general CKLS model,
\[
\Sigma =
\begin{pmatrix}
\int_0^\infty r^{-2\gamma} f(r)\,dr &
-\int_0^\infty r^{1-2\gamma} f(r)\,dr \\[6pt]
-\int_0^\infty r^{1-2\gamma} f(r)\,dr &
\int_0^\infty r^{2-2\gamma} f(r)\,dr
\end{pmatrix}.
\]

Under infill asymptotics, both the ordinary least squares estimator and the Gaussian
quasi–maximum likelihood estimator are consistent for the same true parameter
vector; see \citet{Dehtiar2022} and \citet{Dehtiar2022_CKLS}. Since both
estimators converge in probability to the identical deterministic limit,
their difference must vanish asymptotically. Consequently, the preceding
results apply equally to the likelihood-based estimator.

The preceding asymptotic analysis is derived under the Gaussian increment
structure implied by the pure diffusion specification. This structure ensures
local quadratic approximation of the likelihood and yields the standard
quasi-likelihood estimating equations. In empirical applications, however,
high-frequency financial series often display excess kurtosis and abrupt
movements incompatible with the Gaussian benchmark. Since least squares and
Gaussian quasi-likelihood estimators weight all increments equally, their
finite-sample behavior may be adversely affected by atypically large
increments. This consideration motivates an alternative estimating framework that
retains consistency under the diffusion model while reducing the influence
of extreme observations.

\subsection{Robust Estimation via Density Power Divergence}\label{rob_est_param}

To formalize robustness, we replace the Gaussian quasi-likelihood criterion
by a density power divergence objective in the sense of \citet{Basu1998}.
Let $f_{\boldsymbol{\theta}}$ denote the working Gaussian increment density.
For $\alpha \ge 0$, the density power divergence between the empirical
distribution and $f_{\boldsymbol{\theta}}$ yields the estimating criterion
\[
\mathcal{L}_n^{(\alpha)}(\boldsymbol{\theta})
=
\frac{1}{n}
\sum_{i=1}^n
\left\{
\int f_{\boldsymbol{\theta}}^{1+\alpha}(x)\,dx
-
\left(1+\frac{1}{\alpha}\right)
f_{\boldsymbol{\theta}}^{\alpha}(X_i)+\frac{1}{\alpha}
\right\},
\]
with the usual likelihood recovered at $\alpha=0$.

The resulting minimum density power divergence estimator (MDPDE)
downweights increments with small model-implied density and therefore
limits the influence of extreme observations while preserving
consistency under correct specification.

Let $\theta = (\beta_1,\beta_2,\sigma)^\top$ denote the vector of unknown
parameters, and consider the Gaussian conditional density associated with the
Euler discretization of the diffusion process,
\[
f_\theta(y_t \mid X_t)
=
\frac{1}{\sqrt{2\pi}\sigma}
\exp\!\left(
-\frac{(y_t - \beta_1 z_{1t} - \beta_2 z_{2t})^2}{2\sigma^2}
\right),
\]
where $X_t$ denotes the available conditioning information at time $t$ and
$(z_{1t},z_{2t})$ are the corresponding regressors. This Gaussian approximation
serves as a working model for the conditional distribution of the increments and
constitutes the basis for both classical and robust inference.

For a robustness tuning parameter $\alpha > 0$, the empirical density power
divergence objective function is defined as
\[
H_n^{(\alpha)}(\theta)
=
\sum_{t=1}^n
\left[
\int f_\theta^{1+\alpha}(y \mid X_t)\,dy
-
\frac{1+\alpha}{\alpha}
f_\theta^{\alpha}(y_t \mid X_t)+\frac{1}{\alpha}
\right].
\]
The first term acts as a normalization component that depends only on the model,
while the second term downweights observations that are unlikely under the
postulated conditional density. As $\alpha \downarrow 0$, the criterion
$H_n^{(\alpha)}(\theta)$ converges to the negative log-likelihood, and the
corresponding estimator reduces to the classical maximum likelihood estimator.

The MDPDE is then defined as
\[
\hat{\theta}_n^{(\alpha)}
=
\arg\min_{\theta \in \Theta} H_n^{(\alpha)}(\theta),
\]
where $\Theta$ denotes the parameter space. The tuning parameter $\alpha$
controls the trade-off between robustness and efficiency, with larger values of
$\alpha$ yielding increased resistance to outliers at the cost of some loss in
efficiency under the correctly specified Gaussian model.

\begin{result}[Asymptotic properties of the MDPDE]
\label{lem:mdpde}
Suppose the regularity conditions of \citet{LeeSong2013} hold and the sampling
scheme satisfies the infill asymptotic regime \eqref{eqn:hf}. If in addition, \(
n(\Delta_n)^q \to 0
\quad \text{for some } q>1,
\)
holds, then the minimum density power divergence estimator satisfies
\[
\hat{\theta}_n^{(\alpha)} \xrightarrow{\mathbb{P}} \theta_0,
\]
where $\theta_0$ denotes the true parameter value. Furthermore, if \(
n(\Delta_n)^2 \to 0,
\)
the estimator is asymptotically normal with covariance matrix given explicitly
in \citet{LeeSong2013}.
\end{result}

\begin{remark}
A fundamental property of the minimum density power divergence estimator is its
bounded influence function, which ensures stability in the presence of atypically
large observations. In discretely observed diffusion models, unusually large
increments may arise from unmodeled jump components or mild departures from the
pure diffusion specification. Unlike classical likelihood-based estimators,
whose influence functions are unbounded, the MDPDE automatically downweights such
increments and prevents them from dominating the estimating equations. As a
result, systematic discrepancies between classical and robust estimators provide
a natural diagnostic signal of structural misspecification, including the presence
of jumps.
\end{remark}
This distinction is particularly relevant in high-frequency settings. Diffusion
sample paths are almost surely continuous, and their increments vanish at the
rate $\sqrt{\Delta_n}$ as $\Delta_n \to 0$, whereas jump-induced increments remain
of finite magnitude. Consequently, jump increments appear as local outliers
relative to the continuous diffusion dynamics and may exert disproportionate
influence on likelihood-based procedures, often leading to biased or unstable
estimates. Robust divergence-based estimators mitigate this effect by limiting
the contribution of extreme observations and thereby preserving stability and
interpretability of the fitted model.

An additional advantage arises in constrained parametric models, where extreme
increments may drive likelihood-based estimates toward the boundary of the
parameter space, potentially violating structural conditions such as positivity
or ergodicity. By attenuating the influence of atypical observations, robust
estimation maintains numerical stability and ensures that the estimated parameters
remain within economically and statistically meaningful regimes.
\begin{figure}[htbp]
    \centering
    \begin{subfigure}[t]{0.90\linewidth}
        \centering
        \includegraphics[width=0.9\linewidth]{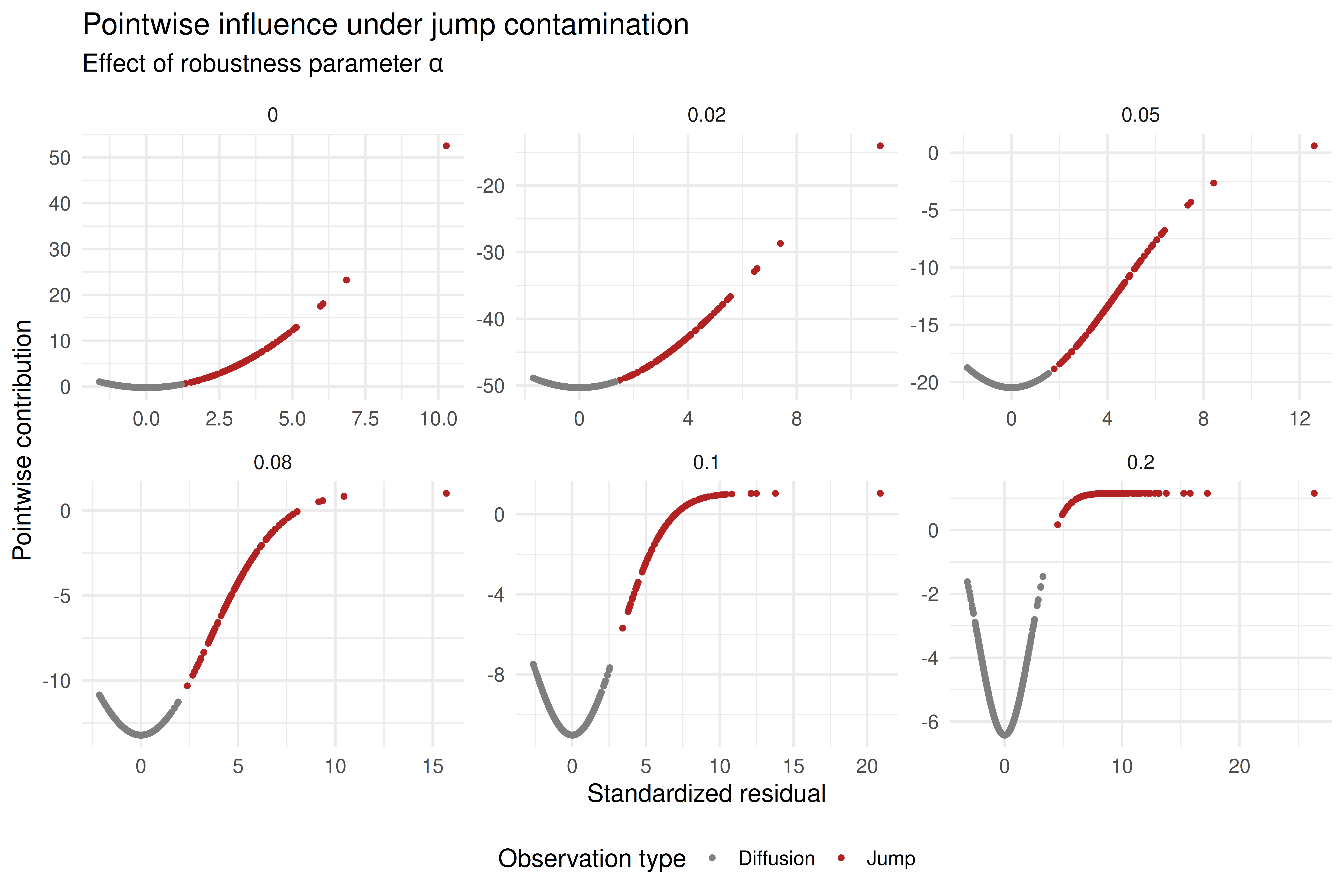}
        \caption{Pointwise contribution (influence) of standardized increments under jump contamination for different values of the robustness parameter $\alpha$. Diffusion-driven increments are shown in gray, while jump-induced increments are shown in red.}
        \label{fig:pointwise_influence}
    \end{subfigure}

    \vspace{0.4cm}

    \begin{subfigure}[t]{0.9\linewidth}
        \centering
        \includegraphics[width=0.9\linewidth]{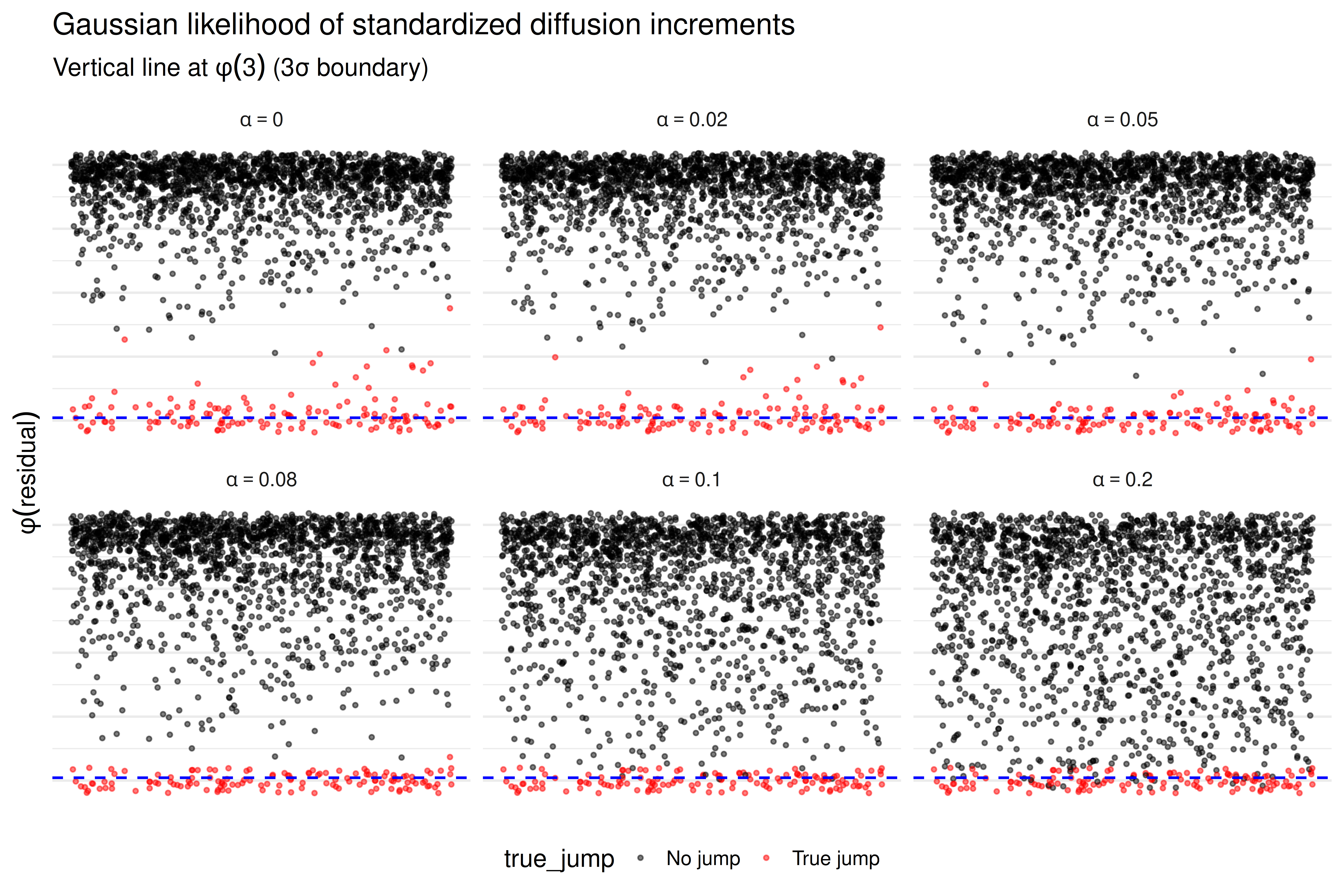}
        \caption{Gaussian likelihood values $\varphi(Residuals)$ of standardized increments with a horizontal reference line at $\varphi(3)$, corresponding to the classical $3\sigma$ boundary, shown for different values of the robustness parameter $\alpha$.}
        \label{fig:likelihood_scatter}
    \end{subfigure}

    \caption{Effect of robustness on pointwise influence and likelihood contributions under jump contamination. Increasing $\alpha$ progressively bounds the influence of large, jump-induced increments while preserving the contribution of diffusion-driven observations.}
    \label{fig:robustness_jump_diagnostics}
\end{figure}
Motivated by these considerations, we consider the CKLS diffusion augmented by a
jump component. Specifically, for known elasticity parameter $\gamma$, we study
the jump--diffusion model
\begin{equation}
\label{eqn:jump_ckls}
dX_t
=
(\beta_1 - \beta_2 X_t)\,dt
+
\sigma X_t^{\gamma}\,dW_t
+
dJ_t,
\end{equation}
where $\{J_t\}_{t \ge 0}$ is a pure-jump process independent of $W_t$. Such models
are widely used in interest rate and volatility modeling to capture abrupt market
movements. The classical Cox--Ingersoll--Ross model arises as the special case
$\gamma = 1/2$.

Figures~\ref{fig:pointwise_influence} and~\ref{fig:likelihood_scatter} illustrate
the effect of robustness on the pointwise contribution of standardized increments.
When $\alpha = 0$, corresponding to Gaussian likelihood estimation, the influence
of extreme increments is unbounded, and jump observations may dominate the
estimating equations. As the robustness parameter $\alpha$ increases, the influence
function becomes bounded, and the contribution of jump-induced increments is
progressively attenuated, while diffusion-driven increments remain largely
unaffected. This behavior reflects the intrinsic robustness of the MDPDE and its
ability to isolate the continuous diffusion structure from jump contamination.

Taken together, these considerations highlight a fundamental advantage of
divergence-based estimation in high-frequency diffusion models. By exploiting the
distinct scaling behavior of diffusion and jump increments, the MDPDE provides
stable parameter estimates and offers a principled basis for detecting departures
from the pure diffusion assumption. This diagnostic capability will play a central
role in the development of the robust parametric jump detection methodology
presented in subsequent sections.
\section{Proposed Test for Jump}\label{test_for_jump}


Let $\{X_{t_i}\}_{i=0}^n$ denote discrete observations of the process introduced
above equation~\eqref{eqn:jump_ckls}, recorded at equidistant times $t_i = i\Delta_n$ over a fixed time horizon,
where $\Delta_n \to 0$. Denote the increments by
\[
\Delta_i^n X := X_{t_i} - X_{t_{i-1}}
\]
and the jump increment be defined as
\begin{equation}
\Delta_i^n J = J_{t_i} - J_{t_{i-1}},
\end{equation}
where the presence of a jump at index $i$ corresponds to the event $\Delta_i J \neq 0$.
Let $\hat{\bm\theta}_n := (\hat\beta_{1n}, \hat\beta_{2n}, \hat\sigma_n)$ denote
consistent estimators of the model parameters satisfying
\[
\hat\beta_{1n} \xrightarrow{P} \beta_1,
\qquad
\hat\beta_{2n} \xrightarrow{P} \beta_2,
\qquad
\hat\sigma_n \xrightarrow{P} \sigma,
\qquad
\Pr(\hat\sigma_n > 0) \to 1.
\]

We define the normalized increment statistic
\begin{equation}
\label{eq:LM_stat}
Z_{i,n}
=
\frac{
\Delta_i^n X
-
(\hat\beta_{1n} - \hat\beta_{2n} X_{t_{i-1}})\Delta_n
}{
\hat\sigma_n X_{t_{i-1}}^\gamma \sqrt{\Delta_n}
}.
\end{equation}

This normalization removes the estimated drift and rescales the increment by the
estimated local diffusion magnitude. Consequently, increments generated by the
continuous martingale component are brought to a unit Gaussian scale, whereas
jump-induced increments remain asymptotically separated due to their larger order.
This separation forms the foundation of the parametric jump detection procedure
developed in the sequel.


The following theorem establishes the asymptotic decomposition of the normalized
increment statistic into its continuous martingale and jump components.

\begin{theorem}
\label{thm:parametric_LM}
Suppose the parameter estimators are consistent. Then, for each fixed index $i$,
\begin{equation}
\label{eq:LM_expansion}
Z_{i,n}
=
\frac{\Delta_i^n W}{\sqrt{\Delta_n}}
+
\frac{\Delta_i^n J}{\sigma X_{t_{i-1}}^\gamma \sqrt{\Delta_n}}
+
o_p(1).
\end{equation}
Consequently, if $\Delta_i^n J = 0$, then
\[
Z_{i,n}
\xrightarrow{\mathcal L}
\mathcal N(0,1),
\]
whereas if $\Delta_i^n J \neq 0$ and the jump magnitude is nondegenerate,
\[
|Z_{i,n}|
\xrightarrow{P}
\infty.
\]
\end{theorem}

Theorem~\ref{thm:parametric_LM} formalizes the pointwise asymptotic separation
between the continuous and discontinuous components. After normalization,
diffusion-driven increments remain stochastically bounded, while jump-induced
increments diverge at rate $\Delta_n^{-1/2}$.

We next characterize the maximal fluctuation of the statistic over intervals
free of jumps, which determines the natural detection boundary.

\begin{theorem}
\label{thm:max_LM}
Let
\[
\mathcal C_n
=
\{ i \in \{1,\dots,n\} : \Delta_i^n J = 0 \}.
\]
Under consistency of the parameter estimators,
\[
\max_{i \in \mathcal C_n}
|Z_{i,n}|
=
\sqrt{2\log(n)}
+
O_p(1),
\]
and more precisely,
\[
\frac{
\max_{i \in \mathcal C_n} |Z_{i,n}| - a_n
}{b_n}
\Rightarrow
\Lambda,
\]
where
\[
a_n
=
\sqrt{2\log n}
-
\frac{\log\log n + \log(\pi)}{2\sqrt{2\log n}},
\qquad
b_n
=
\frac{1}{\sqrt{2\log n}},
\]
and $\Lambda$ denotes the standard Gumbel distribution.
\end{theorem}

This extreme-value characterization determines the intrinsic growth rate of the
continuous-path maximum and provides a principled basis for threshold selection.

\begin{corollary}
\label{cor:uniform_separation}
Under the same assumptions,
\[
\Pr\!\left(
\max_{i \in \mathcal C_n} |Z_{i,n}|
<
\sqrt{2\log(n)}
<
\min_{i:\Delta_i^n J \neq 0} |Z_{i,n}|
\right)
\to 1.
\]
\end{corollary}

Corollary~\ref{cor:uniform_separation} establishes uniform asymptotic separation
between the continuous and jump components, ensuring that a threshold of order
$\sqrt{2\log(1/\Delta_n)}$ asymptotically discriminates between the two regimes.

\begin{theorem}[Consistency of parametric jump detection]
\label{thm:robust_detection_consistency}
Suppose the assumptions of Theorems~\ref{thm:parametric_LM},
\ref{thm:max_LM}, and the consistency of the estimators hold. Then,

\[
\Pr\!\left(
\max_{i \in \mathcal C_n}
|Z_{i,n}| > \xi_n
\right)
\to 0,
\]

and, for any index $i$ such that $\Delta_i^n J \neq 0$,
\[
\Pr\!\left(
|Z_{i,n}| \le \xi_n
\right)
\to 0.
\]

Consequently,
\[
\Pr(\text{correct classification of all increments})
\to 1.
\]
\end{theorem}

Reliable separation requires stable estimation of the diffusion scale parameter.
Classical likelihood-based estimators may be severely distorted by jump-induced
contamination, compromising the normalization. To address this issue, we employ
robust estimators based on the minimum density power divergence.

Let $\hat{\bm\theta}_n^\alpha :=
(\hat\beta_{1n}^\alpha,\hat\beta_{2n}^\alpha,\hat\sigma_n^\alpha)$ denote the robust
estimators, and define the corresponding normalized statistic
\begin{equation}
\label{eq:parm_LM_stat_rob}
Z^\alpha_{i,n}
=
\frac{
\Delta_i^n X
-
(\hat\beta_{1n}^\alpha - \hat\beta_{2n}^\alpha X_{t_{i-1}})\Delta_n
}{
\hat\sigma_n^\alpha X_{t_{i-1}}^\gamma \sqrt{\Delta_n}
}.
\end{equation}

By consistency of the robust estimators, the asymptotic expansion
\eqref{eq:LM_expansion} and the extreme-value result of
Theorem~\ref{thm:max_LM} remain valid for $Z^\alpha_{i,n}$, while providing
enhanced stability under contamination.

We formulate jump identification at index $i$ as the testing problem
\[
H_{0,i} : \Delta_i^n J = 0,
\qquad
H_{1,i} : \Delta_i^n J \neq 0.
\]

Under $H_{0,i}$,
\[
Z^\alpha_{i,n}
\Rightarrow
\mathcal N(0,1),
\]
whereas under $H_{1,i}$ with nonvanishing jump magnitude,
\[
|Z^\alpha_{i,n}|
\xrightarrow{P}
\infty.
\]

Thus, the continuous and discontinuous regimes become asymptotically
distinguishable through the magnitude of the standardized statistic.
Motivated by Theorem~\ref{thm:max_LM}, we introduce the detection threshold
\[
\xi_n
=
\sqrt{2\log n}+c_n,
\qquad
c_n \to \infty,
\]
and declare a jump at time $t_i$ whenever
\begin{equation}
\label{eq:critical_region_annals}
|Z^\alpha_{i,n}|>\xi_n .
\end{equation}

Define the true and estimated jump index sets by
\begin{align}\label{eqn:index_true_est}
    \mathcal J_n
=
\left\{
i\in\{1,\ldots,n\}:\Delta_i^n J\neq 0
\right\},
\qquad
\hat{\mathcal J}_n
=
\left\{
i\in\{1,\ldots,n\}:|Z^\alpha_{i,n}|>\xi_n
\right\}.
\end{align}

Combining the consistency of the robust parameter estimator with
Theorem~\ref{thm:robust_detection_consistency} implies
\[
\Pr\!\left(\hat{\mathcal J}_n=\mathcal J_n\right)\to 1 .
\]
Consequently, conditional on consistent identification of the jump indices,
the jump magnitude at detected times is estimated by
\[
\Delta_i^n\hat J
=
\Delta_i^n X
-
\bigl(\hat\beta_{1n}^{\alpha}
-
\hat\beta_{2n}^{\alpha} X_{t_{i-1}}\bigr)\Delta_n,
\qquad i \in \hat{\mathcal J}_n .
\]

The preceding result establishes asymptotic separation between the
continuous diffusion component and the jump component.
Specifically, the extreme--value growth of the maximal standardized
diffusion increment determines the appropriate detection boundary,
whereas the divergence of the statistic in the presence of jumps ensures
consistent identification of discontinuities.
The robust parametric normalization further stabilizes scale estimation,
thereby enhancing reliability under jump contamination and mild model
misspecification.

\section{Simulation Study}\label{Simu}
We conduct a simulation study to examine the finite-sample behavior of the proposed robust parameter estimation and jump detection procedure. The data generating process is given by \eqref{eqn:jump_ckls}. The jump component $(J_t)$ is specified as a compound Poisson process
\begin{equation}
J_t
=
\sum_{k=1}^{N_t} Y_k,
\end{equation}
where $(N_t)$ is a Poisson process with intensity $\lambda$, and the jump sizes $(Y_k)$ are independent and identically distributed with
\[
Y_k \sim \mathcal{N}(\mu_J,\sigma_J^2).
\]
The process is observed at equidistant time points with sampling interval
\(
\Delta_n = n^{-0.55},
\)
corresponding to a high-frequency asymptotic regime in which
\(\Delta_n \to 0\) as \(n \to \infty\).
The underlying diffusion parameters are set to
\(\beta_1 = 1\), \(\beta_2 = 0.8\), \(\sigma = 0.3\), and \(\gamma = 0.7\).

The simulation study considers the following configurations for the
sample size and jump component parameters:
\begin{align}\label{eqn:parm_set}
n \in \{200,500,1000,1500,2000\}, \quad
\lambda \in \{1,2,3,5\}, \quad
\mu_J \in \{1,2,3,4,5\}, \quad
\sigma_J = 0.1,
\end{align}
in order to evaluate the performance of the proposed robust procedure
for jump isolation and detection.

We first examine performance through a visual comparison between true and
detected jump locations. For illustration, a representative trajectory is
generated under the configuration $(n,\lambda,\mu_J)=(1000,5,3)$.
Using this fixed sample path, jump classification is performed for
robustness parameters
\[
\alpha \in \{0,0.05,0.1,0.15,0.2,0.25\},
\]
where parameter estimation is carried out via the robust procedure and
the resulting estimates are used to classify increments as jump or
diffusion driven. The case $\alpha=0$ corresponds to the classical
least-squares estimator, providing a non-robust benchmark.

Figure~\ref{fig:true_vs_detected} presents the jump detection results for the
values of the robustness parameter $\alpha$ specified above. Rather than
displaying the sample path itself, we plot the first-order increments
$(\Delta_i^n X)$ to highlight discontinuities directly. The detection
threshold obtained from Theorem~\ref{thm:max_LM} is
$\xi_n = 3.512$. True jump times are indicated by blue markers, while
detected jumps based on the thresholded standardized statistic are shown
in red.

When $\alpha=0$, corresponding to the classical OLS estimator, the
procedure fails to identify several clearly separated jumps. As $\alpha$
increases, detection improves substantially, reflecting the stabilizing
effect of robust normalization. The detection rate increases rapidly for
small positive values of $\alpha$ and stabilizes around
$\alpha = 0.15$, beyond which all true jumps are consistently
identified.

These results illustrate the role of the robustness parameter in
controlling the trade-off between sensitivity and stability. Small values
of $\alpha$ leave the normalization susceptible to diffusion-driven
extreme increments, leading to missed or unstable detections, whereas
moderate values attenuate this influence and reduce spurious
classifications while preserving detection power.

\begin{figure}[!ht]
    \centering
    \includegraphics[width=\linewidth]{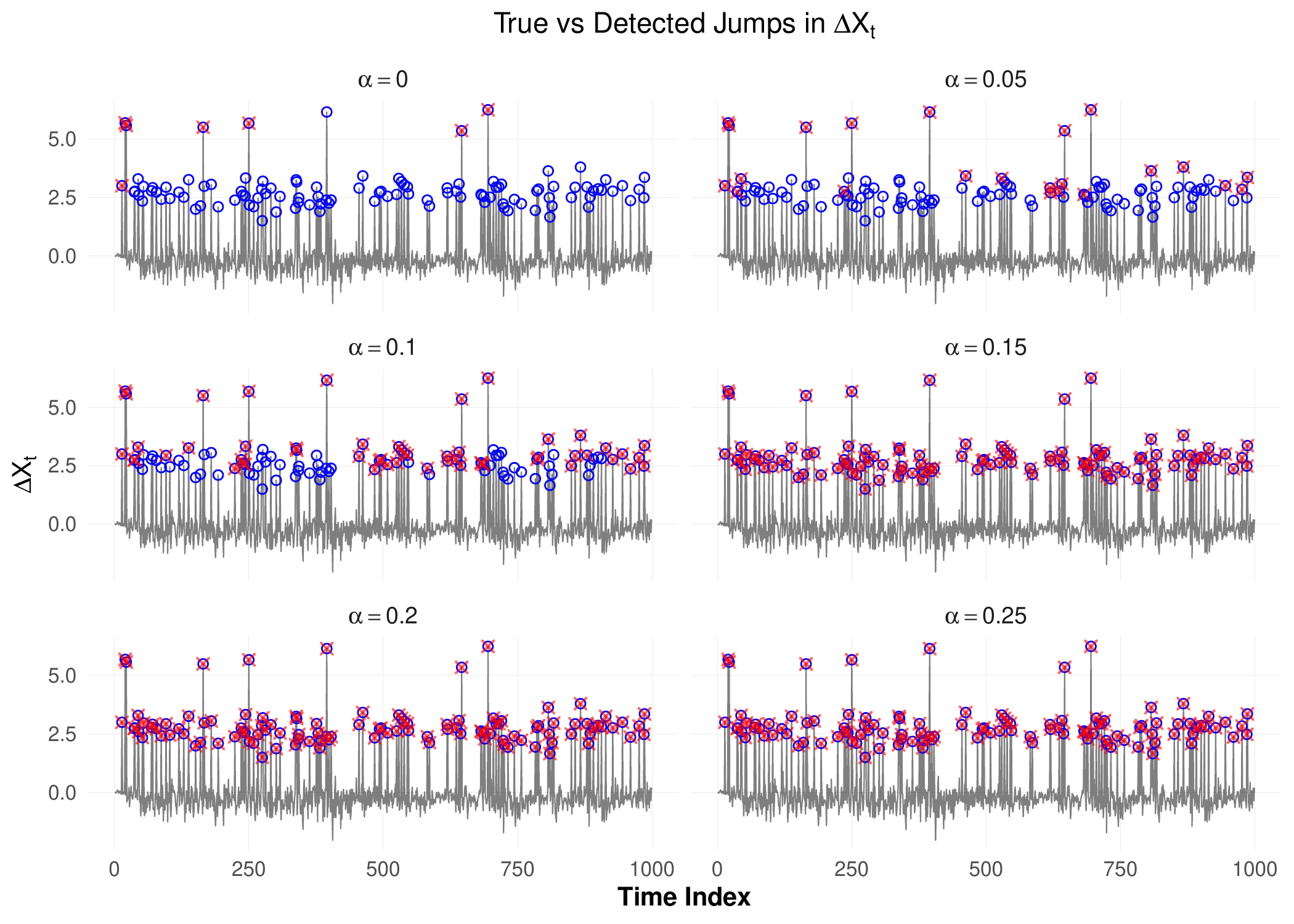}
    \caption{Visual comparison of true and detected jumps in the increment process $\Delta_i^n X$ for different values of the robustness parameter $\alpha$. 
\protect\truejump\ denotes true jumps, whereas \protect\detectedjump\ denotes jumps detected by the proposed procedure.}
    \label{fig:true_vs_detected}
\end{figure}

To provide a more comprehensive assessment of the proposed procedure,
we summarize performance over the parameter configurations specified in
\eqref{eqn:parm_set} for robustness levels
$\alpha \in \secnd{0,0.05,0.1,\dots,0.5}$.
Performance is evaluated using classification accuracy and parameter
recovery metrics, namely the F1-score and a scaled estimation error.

Let $\mathcal J_n$ and $\hat{\mathcal J}_n$ denote the true and estimated
jump sets, respectively, as defined in \eqref{eqn:index_true_est}.
We define the classification counts
\begin{align*}
\mathrm{TP}
&=
|\mathcal J_n \cap \hat{\mathcal J}_n|,
\qquad
\mathrm{FN}
=
|\mathcal J_n \setminus \hat{\mathcal J}_n|,
\qquad
\mathrm{FP}
=
|\hat{\mathcal J}_n \setminus \mathcal J_n|.
\end{align*}
The corresponding performance measures are
\begin{align*}
\text{Precision}
&=
\frac{\mathrm{TP}}{\mathrm{TP}+\mathrm{FP}},
\qquad
\text{Recall}
=
\frac{\mathrm{TP}}{\mathrm{TP}+\mathrm{FN}},\\
\text{F1-score}
&=
2\,\frac{\text{Precision}\times\text{Recall}}
{\text{Precision}+\text{Recall}}.
\end{align*}

Although the jump intensity parameter $\lambda$ is fixed in simulation,
the realized number of jumps varies across trajectories. Consequently,
evaluation based on the theoretical parameters may be misleading in
finite samples. We therefore compare estimators with the
\emph{realized} jump characteristics. Define the realized jump mean and
intensity as
\begin{align}
\tilde{\mu}_J
=
\frac{1}{|\mathcal J_n|}
\sum_{i\in \mathcal J_n} \Delta_i^n J,
\qquad
\tilde{\lambda}
=
\frac{|\mathcal J_n|}{n}.
\end{align}
The corresponding estimators based on detected jumps are
\begin{align}
\hat{\mu}_J
=
\frac{1}{|\hat{\mathcal J}_n|}
\sum_{i\in \hat{\mathcal J}_n} \Delta_i^n \hat J,
\qquad
\hat{\lambda}
=
\frac{|\hat{\mathcal J}_n|}{n}.
\end{align}

When both the sampling size $n$ and the jump intensity are moderate,
the number of realized jumps may be small, preventing reliable
convergence toward theoretical parameters. For this reason, estimation
accuracy is assessed relative to realized quantities rather than
population values. We measure the discrepancy via the scaled error
metric
\begin{align*}
d_M
:=
\sqrt{
\left(\frac{\tilde{\mu}_J}{\hat{\mu}_J}-1\right)^2
+
\left(\frac{\tilde{\lambda}}{\hat{\lambda}}-1\right)^2
},
\end{align*}
which jointly evaluates the accuracy of estimated jump magnitude and
frequency.

Accordingly,  detection performance and estimation accuracy are evaluated
jointly through the F1-score and the scaled error metric defined above.
The F1-score takes values in $[0,1]$, with larger values indicating
improved classification performance; in particular, values closer to one
correspond to more accurate identification of jump indices. 
While the
F1-score assesses detection accuracy through correct classification
counts, the metric $d_M$ quantifies the discrepancy between estimated and
realized jump characteristics.
$d_M$ is scale-free and equals zero if and only if the
estimated quantities coincide with their realized counterparts, that is,
$\hat{\mu}_J=\tilde{\mu}_J$ and $\hat{\lambda}=\tilde{\lambda}$.
Consequently, smaller values of $d_M$ indicate greater estimation
accuracy, whereas larger values reflect increasing deviation from the
realized jump structure and hence reduced reliability of both parameter
estimation and detection. In this sense, the F1-score captures
classification performance, while $d_M$ measures parametric divergence,
providing complementary perspectives on the effectiveness of the proposed
procedure.

Figure~\ref{fig:F1_rate} reports the F1-score across the parameter
configurations defined in~\eqref{eqn:parm_set}. The results show a clear
improvement in detection performance as the jump mean increases.
Moving from left to right in the figure, corresponding to larger values
of $\mu_J$, the F1-score increases monotonically and approaches one,
indicating nearly perfect recovery of the jump indices.
A similar improvement is observed when the sample size increases
(top to bottom panels). Larger samples lead to a faster convergence of
the F1-score toward one as the robustness parameter $\alpha$ increases.

For small sample sizes and higher jump intensities (large $\lambda$),
the number of jumps is relatively large compared with the available
observations, which reduces the power to correctly identify all true
jumps. Nevertheless, relative to the case $\alpha=0$ (corresponding to
the OLS estimator), the detection accuracy improves substantially for
$\alpha>0$, demonstrating the advantage of the proposed robust
normalization. These findings are consistent with the theoretical
result in Theorem~\ref{thm:robust_detection_consistency}, which implies
that diffusion increments remain of stochastic order
$\sqrt{\Delta_n}$ while jump-induced increments diverge, yielding
increasing separation as $\Delta_n \to 0$.

A similar pattern is observed in the jump-parameter estimation results
shown in Figure~\ref{fig:metric_conv}. As the separation between the
diffusion increments and the realized jump mean $(\tilde{\mu}_J)$
increases (again moving from left to right), the proposed metric
converges more rapidly toward zero, indicating improved estimation
accuracy. When the jump intensity is very small ($\lambda=1$), the
performance of the OLS estimator is nearly indistinguishable from that
of the robust estimator with $\alpha>0$. In contrast, for higher jump
intensities and small samples, the discrepancy between the estimated
and realized parameters becomes more pronounced. Increasing the sample
size substantially reduces this discrepancy, and the metric declines
more steeply toward zero, reflecting improved estimation precision.

\begin{figure}[H]
    \centering
    \includegraphics[width=0.9\linewidth]{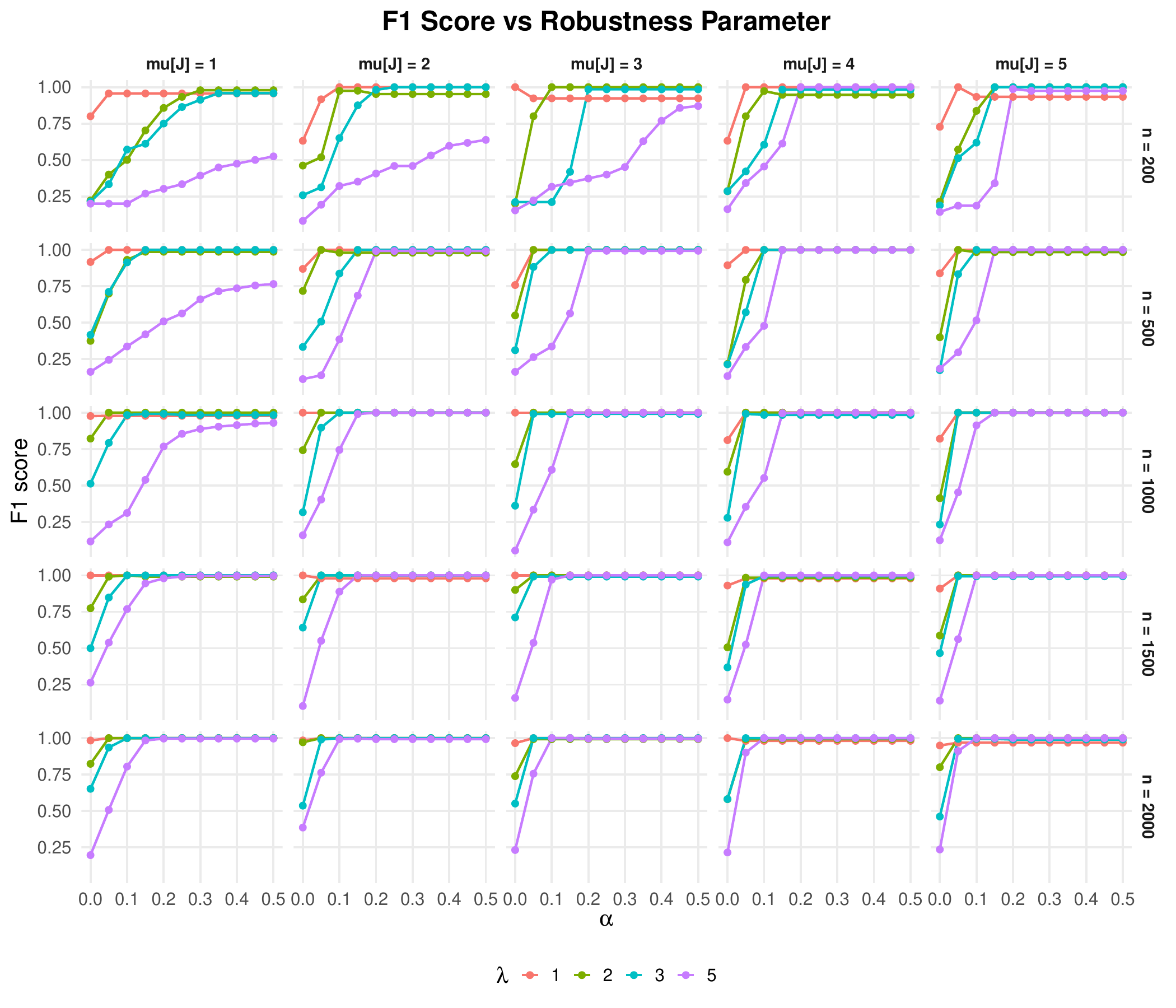}
    \caption{F1 score as a function of the robustness parameter $\alpha$ for different sample sizes ($n$), jump intensities ($\lambda$), and jump mean values ($\mu_J$). Panels are arranged by sample size (rows) and jump mean (columns), while curves correspond to different jump intensities. The figure illustrates how the robustness parameter affects jump detection accuracy and how this effect varies with jump magnitude, jump frequency, and sampling resolution.}
    \label{fig:F1_rate}
\end{figure}
Overall, the simulation results provide empirical support for the
theoretical separation mechanism underlying the proposed detection
procedure. The robust parametric normalization preserves the
asymptotic divergence of jump-induced increments while stabilizing the
estimation of the diffusion scale. As a consequence, the resulting
procedure achieves reliable identification of discontinuities across a
broad range of jump intensities and magnitudes, maintaining stable
performance even under moderate contamination through jumps.

\begin{figure}[H]
    \centering
    \includegraphics[width=0.9\linewidth]{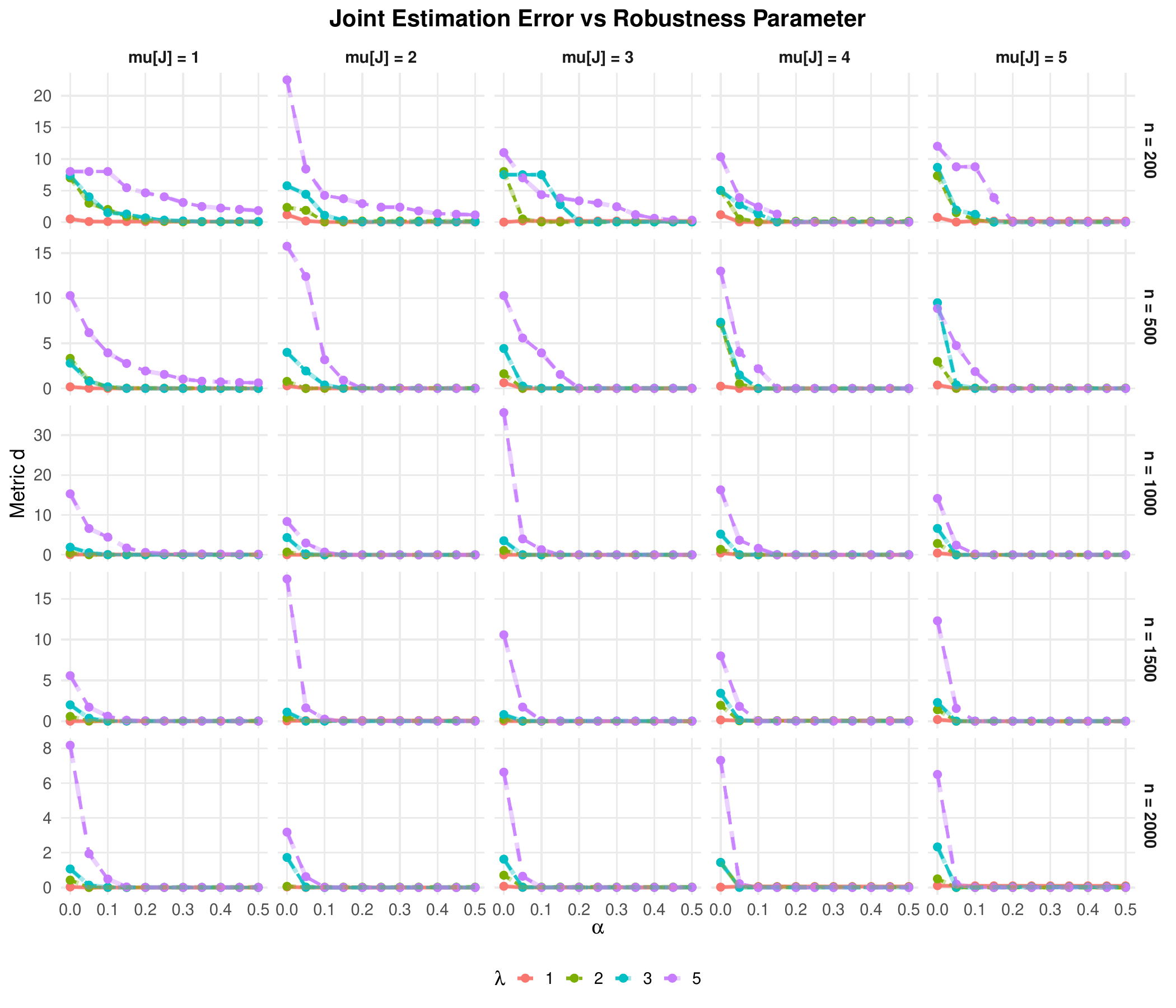}
    \caption{Error metric ($d_M$) as a function of the robustness parameter $\alpha$ under varying sample sizes ($n$), jump intensities ($\lambda$), and jump mean values ($\mu_J$). Panels are arranged by sample size (rows) and jump mean (columns), while curves correspond to different jump intensities. The figure illustrates how the robustness tuning parameter influences the accuracy of jump parameter estimation and how this effect interacts with jump magnitude, jump frequency, and sampling resolution.}
    \label{fig:metric_conv}
\end{figure}




\section{Conclusion}\label{concl}

This paper develops a robust estimation and jump detection framework for discretely observed CKLS jump–diffusion processes under high–frequency asymptotics. The methodology combines a density power divergence–based estimator for drift parameters with an extreme–value–theoretic jump detection rule derived from standardized residuals. The proposed approach is designed to retain asymptotic validity under the diffusion benchmark while mitigating the influence of atypical increments arising from jumps or other departures from Gaussianity.

On the theoretical side, we establish consistency and asymptotic normality of the robust estimator under infill asymptotics and show that the associated residual normalization yields a detection statistic whose extremal behavior permits asymptotically valid separation of continuous and discontinuous components. The analysis demonstrates that robustness regularizes the normalization without altering the asymptotic classification boundary, thereby preserving identifiability of the jump component.

The finite–sample experiments corroborate the theoretical findings. The results show that moderate robustness improves stability of both parameter estimation and jump identification, particularly in regimes where jump magnitudes are small or sample sizes are limited. As the sampling frequency increases, detection performance approaches the asymptotic regime, confirming the theoretical scale separation between diffusion and jump increments.

Overall, the proposed framework provides a unified and theoretically grounded approach to robust inference in jump–diffusion models. It achieves reliable jump detection and stable parameter estimation across a broad range of regimes, while remaining consistent with classical likelihood–based inference under the pure diffusion specification. These results support the use of robustness as a principled mechanism for improving finite–sample performance without compromising asymptotic efficiency.

\section*{Acknowledgments}

The author would like to express special thanks to Prof. Diganta Mukherjee \footnote{Professor, Indian Statistical Institute} for his valuable guidance and insightful suggestions throughout this research. The author is also grateful to Prof. Indranil Sengupta\footnote{Professor, City University of New York} for his assistance and support during the development of this work. Their contributions have been instrumental in shaping the final outcome of this paper.

\printbibliography

\appendix
 \section{Appendix}\label{sec:app}

\setcounter{subsection}{0}
\renewcommand{\thesubsection}{A\arabic{subsection}}
\renewcommand{\theequation}{\thesection.\arabic{subsection}.\arabic{equation}}
\numberwithin{equation}{subsection}

\begin{proof}[Proof of the theorem \ref{thm:parametric_LM}]
Integrating \eqref{eqn:jump_ckls} over $(t_{i-1},t_i]$ yields
\[
\Delta_i^n X
=
\int_{t_{i-1}}^{t_i}
\kappa(\theta-X_s)\,ds
+
\int_{t_{i-1}}^{t_i}
\sigma{X_s^\gamma}\,dW_s
+
\Delta_i^n J.
\]
A first-order It\^o--Taylor expansion gives
\[
\Delta_i^n X
=
\kappa(\theta-X_{t_{i-1}})\Delta_n
+
\sigma{X_{t_{i-1}}^\gamma}\Delta_i^n W
+
\Delta_i^n J
+
O_p(\Delta_n).
\]
Subtracting the estimated drift component and invoking consistency of the parameters yields
\[
\Delta_i^n X
-
\hat\kappa_n(\hat\theta_n-X_{t_{i-1}})\Delta_n
=
\sigma X_{t_{i-1}}^\gamma\Delta_i^n W
+
\Delta_i^n J
+
o_p(\sqrt{\Delta_n}).
\]
Since
\[
\hat\sigma_nX_{t_{i-1}}^\gamma\sqrt{\Delta_n}
=
\sigma X_{t_{i-1}}^\gamma \sqrt{\Delta_n}(1+o_p(1)),
\]
division yields \eqref{eq:LM_expansion}, from which the stated conclusions follow.
\end{proof}
\begin{proof}[Proof of Theorem~\ref{thm:max_LM}]
Let
\[
M_n := \max_{i\in\mathcal C_n} |Z_{i,n}|.
\]
By consistency of the parameter estimators,
\[
Z_{i,n}
=
\frac{\Delta_i^n X - b(X_{t_{i-1}})\Delta_n}
{\sigma(X_{t_{i-1}})\sqrt{\Delta_n}}
+ o_p(1)
=
\xi_i + o_p(1),
\]
uniformly over $i\in\mathcal C_n$, where $\{\xi_i\}$ are independent $N(0,1)$ variables. Since $|\mathcal C_n|/n \xrightarrow{P} 1$, it follows that
\[
M_n
=
\max_{1\le i\le n} |\xi_i|
+ o_p(1).
\]

Thus, for any $u>0$,
\[
\Pr(M_n \le u)
=
\left[2\Phi(u)-1\right]^n + o(1).
\]
Define the normalizing sequences
\[
a_n
=
\sqrt{2\log n}
-
\frac{\log\log n + \log\pi}{2\sqrt{2\log n}},
\qquad
b_n
=
\frac{1}{\sqrt{2\log n}},
\]
and let $u_n(x)=a_n+b_n x$. Using the Gaussian tail expansion
\[
1-\Phi(u)
=
\frac{\phi(u)}{u}(1+o(1)),
\qquad u\to\infty,
\]
we obtain
\[
n\left[1-\Phi(u_n(x))\right]
\to e^{-x}.
\]
Consequently,
\[
\Pr\!\left(
\frac{M_n-a_n}{b_n} \le x
\right)
\to
\exp(-e^{-x}),
\]
which establishes convergence to the standard Gumbel law.

Finally, since $(M_n-a_n)/b_n=O_p(1)$ and $a_n=\sqrt{2\log n}+o(1)$,
\[
M_n
=
\sqrt{2\log n}
+
O_p(1).
\]
\end{proof}

\begin{proof}[Proof of Theorem~\ref{thm:robust_detection_consistency}]
Let
\[
\mathcal C_n = \{i : \Delta_i^n J = 0\},
\qquad
\mathcal J_n = \{i : \Delta_i^n J \neq 0\},
\]
denote the sets of continuous and jump increments, respectively, and define the threshold
\[
\xi_n = a_n + b_n G^{-1}(1-q),
\]
where $a_n,b_n$ are given in Theorem~\ref{thm:max_LM} and $G^{-1}$ is the Gumbel quantile function.

We first control false detections. By Theorem~\ref{thm:max_LM},
\[
\frac{\max_{i\in\mathcal C_n}|Z_{i,n}| - a_n}{b_n}
\Rightarrow \Lambda,
\]
where $\Lambda$ has the standard Gumbel distribution. Hence,
\[
\Pr\!\left(\max_{i\in\mathcal C_n}|Z_{i,n}| > \xi_n\right)
=
\Pr\!\left(\Lambda > G^{-1}(1-q)\right) + o(1)
=
q + o(1).
\]
Since $q>0$ may be chosen arbitrarily small, false detections vanish asymptotically.

Next, consider detection under the alternative. For any fixed $i\in\mathcal J_n$, Theorem~\ref{thm:parametric_LM} yields
\[
|Z_{i,n}| \xrightarrow{P} \infty.
\]
Since $\xi_n = O(\sqrt{\log n})$, it follows that
\[
\Pr(|Z_{i,n}| \le \xi_n) \to 0,
\]
so each jump is detected with probability tending to one.

Finally, define the estimated jump set
\[
\widehat{\mathcal J}_n = \{i : |Z_{i,n}| > \xi_n\}.
\]
The preceding bounds imply
\[
\Pr\!\left(\widehat{\mathcal J}_n = \mathcal J_n\right) \to 1,
\]
which establishes consistency of the classification.
\end{proof}


\clearpage
\setcounter{section}{0}
\setcounter{equation}{0}
\setcounter{figure}{0}
\setcounter{table}{0}

\renewcommand{\thesection}{S.\arabic{section}}
\renewcommand{\thesubsection}{S.\arabic{section}.\arabic{subsection}}

\renewcommand{\theequation}{S.\arabic{equation}}

\renewcommand{\thefigure}{S.\arabic{figure}}
\renewcommand{\thetable}{S.\arabic{table}}


\end{document}
\section*{Supplementary Material}
\include{table_tex_files/jump_table_n200_lambda1}
\include{table_tex_files/jump_table_n200_lambda2}
\include{table_tex_files/jump_table_n200_lambda3}
\include{table_tex_files/jump_table_n200_lambda5}

\include{table_tex_files/jump_table_n500_lambda1}
\include{table_tex_files/jump_table_n500_lambda2}
\include{table_tex_files/jump_table_n500_lambda3}
\include{table_tex_files/jump_table_n500_lambda5}

\include{table_tex_files/jump_table_n1000_lambda1}
\include{table_tex_files/jump_table_n1000_lambda2}
\include{table_tex_files/jump_table_n1000_lambda3}
\include{table_tex_files/jump_table_n1000_lambda5}

\include{table_tex_files/jump_table_n1500_lambda1}
\include{table_tex_files/jump_table_n1500_lambda2}
\include{table_tex_files/jump_table_n1500_lambda3}
\include{table_tex_files/jump_table_n1500_lambda5}

\include{table_tex_files/jump_table_n2000_lambda1}
\include{table_tex_files/jump_table_n2000_lambda2}
\include{table_tex_files/jump_table_n2000_lambda3}
\include{table_tex_files/jump_table_n2000_lambda5}

\end{document}

https://archive.org/details/g.f.simmonsdifferentialequations/page/n7/mode/2up